\documentclass[review]{elsarticle}
\usepackage{color}
\usepackage{algorithm}
\usepackage[noend]{algpseudocode}

\makeatletter
\def\BState{\State\hskip-\ALG@thistlm}
\makeatother

\def\*#1{\mathbf{#1}}
\def\~#1{\boldsymbol{#1}}

\usepackage{listings}
\usepackage{color} %red, green, blue, yellow, cyan, magenta, black, white
\definecolor{mygreen}{RGB}{28,172,0} % color values Red, Green, Blue
\definecolor{mylilas}{RGB}{170,55,241}
\usepackage[margin=0.5in]{geometry}
\usepackage{amsmath}
\usepackage{amssymb}
\usepackage{graphicx}
\usepackage{graphicx}
\usepackage{subcaption}

\usepackage{amsfonts, amsmath, bm, bbm, amsthm} 
\newtheorem{theorem}{Theorem}[section]
\newtheorem{lemma}[theorem]{Lemma}
\newtheorem{remark}[theorem]{Remark}
\newtheorem{assumption}[theorem]{Assumption}
\newtheorem{definition}[theorem]{Definition}

\begin{document}
\begin{frontmatter}
%%%%% title : short title may not be used but TITLE is required.
% \title{TITLE}
% \title[short title]{TITLE}
\title{Interpretable Selective Learning in Credit Risk}
\author{Dangxing Chen \thanks{Corresponding author. Email: dangxing.chen@dukekunshan.edu.cn}}
\author{Weicheng Ye
% \thanks{Most work was done while at Carnegie Mellon University.}
}
\author{Jiahui Ye
% \thanks{Most work was done while at Carnegie Mellon University.}
}
% \thanks{Most work was done while at Carnegie Mellon University. Opinions expressed in this paper are those of the authors, and do not necessarily reflect
% the view of Credit Suisse.}
% \affiliation[1]{Zu Chongzhi Center for Mathematics and Computational
%Sciences, Duke Kunshan University, Kunshan, Jiangsu, China}
%\affiliation[2]{Shanghai, China}
\date{}
% \maketitle

% \date{}
% \maketitle

\begin{abstract} 

% The forecasting of the credit default risk has been an important research field for several decades. Traditionally, linear models, e.g., logistic regression, have been widely recognized as a solution due to their accuracy and interpretability. As a recent trend, a tremendous number of researchers have tended to use more complex and advanced machine learning methods to improve the accuracy of the prediction. Although certain non-linear machine learning methods have better predictive power, they are often considered as lack of interpretability by financial regulators. Hence, they have not been widely used in credit risk assessment. To enhance the interpretability, we introduce a neural network with the selective option to distinguish whether the datasets can be explained by the linear models or not. For the portion of the datasets that cannot be well supported by the linear model, our learning model can feed the data into  deeper machine learning methods. Accordingly, we find that, for most of the datasets, logistic regression will be sufficient, with reasonable accuracy; meanwhile, for some specific data portions, a shallow neural network model leads to much better accuracy without significantly sacrificing the interpretability. 

The forecasting of the credit default risk has been an important research field for several decades. Traditionally, logistic regression has been widely recognized as a solution due to its accuracy and interpretability. As a recent trend, researchers tend to use more complex and advanced machine learning methods to improve the accuracy of the prediction. Although certain non-linear machine learning methods have better predictive power, they are often considered to lack interpretability by financial regulators. Thus, they have not been widely applied in credit risk assessment. We introduce a neural network with the selective option to increase interpretability by distinguishing whether the datasets can be explained by the linear models or not. We find that, for most of the datasets, logistic regression will be sufficient, with reasonable accuracy; meanwhile, for some specific data portions, a shallow neural network model leads to much better accuracy without significantly sacrificing the interpretability.

\end{abstract}

\end{frontmatter}

\section{Introduction}

% \dan{
% Credit scoring is concerned with developing empirical models to support decision-making in the retail credit business \cite{crook2007recent}. This sector is of considerable economic importance. For example, the value of consumer credit outstanding in the United States is $\$$4,168.43 billion in 2020 \footnote{https://www.statista.com/statistics/188170/consumer-credit-liabilities-of-us-households-since-1990/}.  
% This figure indicates that financial institutions require formal tools to inform lending decisions. 
% }

Understanding the credit risk and properly managing the risk has always been a hot topic in the financial industry. By developing reliable credit scoring methods using empirical models, lenders and financial institutions are able to estimate the risk levels and make risk-based decisions to properly hedge the risk of default. More importantly, credit risk has a considerable economic impact globally. For instance, the value of customer credit outstanding in the United States was USD 4,168.43 billion in 2020 (https://www.statista.com/statistics/188170/consumer-credit-liabilities-of-us-households-since-1990/), which was nearly 20 percent of the US GDP for the same year (https://fred.stlouisfed.org/series/GDPA). Both the financial and economic impacts require the lenders and financial institutions to carefully choose risk assessment methods and precisely gauge the credit risks to avoid extreme situations due to underestimated credit risks. 

Credit scoring is a universal risk assessment method leveraging statistical models to determine the credit worthiness of a borrower. In other words, a credit score is a model-based estimate of the probability used by lenders to understand the likelihood that a borrower will default soon. Since defaulting or not defaulting is a binary outcome, researchers and lenders commonly use classification algorithms to estimate the default probability of the borrowers \cite{hand1997statistical}. Among all the classification algorithms, logistic regression is the most popular method in the industry because of its good predictive power and, more importantly, its simplicity. On the other hand, researchers such as our group are continuously working on exploring more complicated methods, such as machine learning, for higher accuracy and better interpretability to help improve the risk management process.

Machine learning methods automatically learn from the data and improve from previous experience. They have been widely used in applications such as computer vision, pattern recognition, and text learning. 
Recently, there have been several efforts made by researchers to apply machine learning methods to credit scoring. Some literature reviews can be found in \cite{baesens2003benchmarking, lessmann2015benchmarking,yeh2009comparisons}.
To mention a few, some techniques have been used in earlier works, including decision trees \cite{srinivasan1987credit}, k-nearest neighbors \cite{henley1997construction}, neural network \cite{yobas2000credit}, and support vector machines \cite{baesens2003benchmarking}. More recently, the adoption of ensemble methods has also shown a great improvement in terms of accuracy \cite{finlay2011multiple,lessmann2015benchmarking,paleologo2010subagging}. Machine learning methods in general have attracted considerable attention from the credit industry \cite{grennepois2018using}. 

Machine learning methods have great flexibility to deal with high-dimensional and highly non-linear datasets. They have improved the accuracy of predicting the probability of default, when properly used \cite{finlay2011multiple}. However, machine learning methods are usually restricted to a black box and it is very difficult to interpret the result. 
This is one of the most significant challenges that researchers and the credit industry are facing because decisions regarding credit applications cannot be made based on discretion.
Financial regulators have enforced the reasoning of institutional and individual credit decisions. 
With the new General Data Protection Regulation (GDPR) \cite{voigt2017eu}, including the ``right to an explanation,” explanations must be provided to justify the application decisions. 
In particular, if an application for a credit card is rejected, \text{justification of the rejection must be provided to borrowers and regulators.} A black-box machine learning technique can be hardly accepted without explanations. 
Hence, there has been an increasing trend in the machine learning community to improve the interpretability of the machine learning models \cite{chen2018interpretable,dash2018boolean,gomez2020vice,horel2020significance,dumitrescu2022machine}.

% During the long history, self-awareness has been developed by human to avoid potential unknown risks. In statistics, an equivalent idea has been provided to improve decision making. Such idea is referred to selective prediction, which is also known as prediction with a rejection option. 
During our research, we have found that we could use a selective framework choosing between traditional credit scoring models and machine learning methods to improve the interpretability. In statistics, there is a term called prediction with a rejection option. It is introduced as selective prediction, similar to the self-awareness of knowing what we do not know.
The concept of the reject option can be traced back to Chow \cite{chow1970reject}, and has been extensively studied for various hypothesis classes and learning algorithms, such as support vector machine, boosting, and nearest neighbors \cite{fumera2002,hellman1970,cortes2016}.
More recently, Geifman and El-Yaniv extended the concept of the reject option to neural networks, which could leverage the neural network's advantage in data fitting and error-reject trade-off \cite{elyaniv2019selnet}. 
They introduced Selective Net, a neural network embedded with a reject option, which allows the end-to-end optimization of selective models \cite{elyaniv2019selnet}.
Selective Net has also been extended to have an adaptive rejection option with the optimal rejection threshold searching \cite{ye2022}.
The main motivation for selective prediction is to reduce the error rate by abstaining from prediction while keeping coverage as high as possible. Due to the uncertainty in image classification and pattern recognition, the selective prediction has attracted consistent interest in practice \cite{xie2006bootstrap, santos2005optimalroc}.
In many mission-critical machine learning applications, such as autonomous driving,
medical diagnosis, and home robotics, detecting and controlling statistical uncertainties in machine learning processes is essential. 
These AI tasks can benefit from effective selective prediction. 
For example, if a self-driving car can identify a situation in which it does not know how to respond, it can alert the human drivers to take over and the risk is thus controlled \cite{yaniv2017}. Therefore, besides accuracy improvement, the selective option has provided new insights for knowledge discovery, which could potentially benefit financial applications.

In this paper, we ask the following question: if a machine learning method could outperform logistic regression, which part of the dataset performs better in prediction? 
To answer the above question, we utilized the novel idea of a selective learning framework, bridging the gap between well-understood logistic regression and a black-box neural network.
Here, we provide the reasoning: although the majority of the credit rating datasets are able to be explained by well-learned linear models due to their good interpretability, a small portion of the datasets still contain some non-linear patterns that should be fed into deeper machine learning methods for further study, as the linear model cannot explain this small portion well.
Non-linearity effects, such as the diminishing marginal effect, are commonly observed in credit scoring datasets. For example, for a borrower with a perfect credit record, missing a payment significantly increases his/her chance of defaulting. On the other hand, allowing five late payments does not lead to such a sharp increase. At a certain point, more missing payments do not guarantee a significantly higher probability of default. Thus, the marginal probability diminishes and therefore it is non-linear. 
If our method can identify the non-linear portion of samples, detailed explanations can be provided to regulators and borrowers to illustrate the usage of machine learning methods.
Hence, we believe that the idea of selective prediction naturally fits with credit scoring.
% In this study, we want to give our learning model a rejection option to the dataset portion that is unexplained well by the linear classifier.

Our selective learning framework incorporates the recent idea of selective options into the traditional machine learning methods. Different from Selective Net, our framework aims to improve the interpretability instead of the accuracy of the neural network. 
In this work, by comparing empirical performance between logistic regression and a neural network, a novel selective labeling technique is introduced to separate the dataset into linear and non-linear parts, where the non-linear parts are considered as an equivalent version of the rejected set. 
Then, a Difference Net is used to train the selective labels. As a byproduct, the rejection rate of the dataset can be accurately estimated. The Difference Net learns the improvement of the neural network over logistic regression, delivering explanations to regulators/customers/loan officers, and thus disentangling the black-box structure of the neural network. One additional advantage of our framework is that minimal modifications are made for the traditional logistic regression approach, providing a smooth transition to machine learning methods. 
In this paper, rigorous theoretical justifications are provided to support our argument. 

We conduct an extensive empirical investigation on the selective learning framework using different data sources. Most of the samples in the dataset can be well explained by logistic regression, as suggested by the low rejection rate. Difference Net successfully identifies the weakness of logistic regression, where strong non-linearity is observed. For the rejected set, the neural network significantly outperforms logistic regression. In particular, risks are notably underestimated by logistic regression, due to its failure to capture non-linear effects, such as diminishing marginal effects of features. Finally, detailed comprehensive explanations are provided to meet explanation requirements. 

The contributions of this work are as follows: 
(i) the selective learning framework, which builds the bridge between transparent logistic regression and a highly accurate black-box neural network for credit risk; (ii) the detailed interpretations of Difference Net from a different perspective to satisfy the requirements of regulators, managers, data professionals, and borrowers; (iii) the rigorous theoretical justification of our framework. 

% In this paper, we ask the following question: If a machine learning method could improve the accuracy over the logistic regression, which part of the dataset is improved? To answer the question, we utilized a novel idea of the selection option in the neural network framework. Introduced by Geifman et al., incorporating the selection option into the neural network has provided great interpretability to the application \cite{geifman2019selectivenet,geifman2017selective}. In their original papers, a reject option is provided to the neural network to identify unexplainable datasets by the neural network. Here, we utilize the selection idea and use the neural network to select the data that cannot be explained by the logistic regression, called the Selective Net. The major advantages of our approach are two folds: 1. The structure of the logistic regression is not altered. For most of the dataset, the results from the logistic regression can still be accepted, which meets the existing requirements by regulators; 2. A small part of the dataset failed to be explained by the logistic regression, in which the neural network has significant improvements over the logistic regression. Special attention should be paid to this part of the dataset. Finally, further analysis could be applied to understand the patterns in such datasets to better understand the data. 

The rest of the paper is organized as follows: we introduce our methodology in Section 2. In Section 3, the empirical results are presented. We conclude the article in Section 4.

\section{Methodology}

\subsection{Two-stage learning}

Assume that we have the joint data and corresponding label space $\mathcal{D} \times \mathcal{Y}$, where $\mathcal{D}$ is the dataset that has $n$ total samples, with $p$ as the number of features and $\mathcal{Y}$ as the corresponding labels. 
Here, we assume that the data generating process follows the assumption given below. 
\begin{assumption}\label{assump:model}
At each sample point $\*x$, the default event follows the binomial distribution with the probability $p(\*x)$, where $p(\*x)$ is continuous. 
\end{assumption}

As described earlier, we consider two models here: logistic regression and neural network. 
Due to its capacity for fitting complicated, high-dimensional, non-linear functions, the neural network has been widely studied and discussed in both academia and industry. 
The approximation power of the neural network can be summarized by the universal approximation theorem \cite{cybenko1989approximation,hornik1991approximation,kubat1999neural,hassoun1995fundamentals}.
\begin{theorem}[Universal Approximation Theorem] \label{theorem:universal}
Fix a continuous function $\sigma :\mathbb {R} \rightarrow \mathbb {R}$ (activation function) and positive integers $d,D$. The function $\sigma$  is not a polynomial if, and only if, for every continuous function $f: \mathbb {R} ^{d} \to \mathbb {R} ^{D}$ (target function), every compact subset $K$ of $\mathbb {R} ^{d}$, and every $\epsilon >0$, there exists a continuous function $f_{\epsilon}: \mathbb {R} ^{d}\to \mathbb {R} ^{D}$ (the layer output) with representation
\begin{align*}
f_{\epsilon }=W_{2}\circ \sigma \circ W_{1},
\end{align*}
where $W_{2},W_{1}$ are composable affine maps and $\circ$  denotes the component-wise composition, such that the approximation bound
\begin{align*}
\sup _{x\in K}\,\|f(x)-f_{\epsilon }(x)\|< \epsilon
\end{align*}
holds for any $\epsilon$ that is arbitrarily small.
\end{theorem}
Under Assumption~\ref{assump:model}, the Universal Approximation Theorem guarantees that the target function $p(\*x)$ can be modeled by the neural network accurately and with the appropriate effort.
Therefore, we consider the neural network as a ground truth with the output function $f(\*x)$. For prediction, a threshold $\tau$ is used such that 
\begin{align*}
F(\*x) = 
\begin{cases}
1, & \text{if $f(\*x) \geq \tau$}, \\
0, & \text{otherwise}.
\end{cases}
\end{align*}
Here, the choice of $\tau$ will be determined by credit rating institutions with domain expertise. 
For simplicity, we use $\tau=0.5$ unless otherwise mentioned. All methods in this paper use the same choice of $\tau$.

While the neural network has good generalization, its black-box structure poses a challenge to the interpretability. 
Lack of interpretation has prevented the usage of neural networks in the financial industry, as explanations are required from regulators. 
Traditionally, logistic regression has been the industry standard for the credit scoring problem, due to its accuracy and, most importantly, interpretability.
In machine learning, a logistic regression model can simply be regarded as a special case of the neural network with no hidden layers; therefore, it serves as a coarse model of a neural network. 
We assume that logistic regression has the output $\widetilde{f}(\*x)$ with prediction $\widetilde{F}(\*x)$, where the logit function is assumed to be linear with respect to all variables,
\begin{align*}
\ln \left( \frac{ \widetilde{f}(\*x) }{ 1-\widetilde{f}(\*x) } \right) = \sum_{i=1}^p a_i x_i.     
\end{align*}
During the first stage, the logistic regression and neural network are fitted by the datasets separately. 

We wish to build a model with logistic regression's interpretability and the neural network's generalization power. 
To achieve this goal, we introduce an additional layer to build the bridge between the logistic regression and neural network. We employ another neural network to learn the difference between two models. 
We consider a Difference Net with the output $g(\*x)$ and prediction $G(\*x)$ that serves as a binary qualifier for $F$ and $\widetilde{F}$ to determine whether the data can be explained by the logistic regression. For sample $\*x \in \mathcal{D}$, 
\begin{align}
G(\*x) = \begin{cases}
1, &\text{if $\*x$ is accepted by the logistic regression}; \\
0, &\text{if $\*x$ is rejected by the logistic regression}.
\end{cases}
\end{align} 
The Difference Net is learned from the empirical data. 
Labels are required in order to properly train the net. 
Ideally, we consider a new selective labeling $\mathcal{Z}$ such that
\begin{align}
z = \begin{cases}
1, &\text{if $F(\*x) = \widetilde{F}(\*x)$},  \\
0, &\text{if $F(\*x) \neq \widetilde{F}(\*x)$}.
\end{cases}
\end{align}
We name the data with $z=0$ as our reject portion of the datasets.
In practice, however, there is the possibility that the neural network may not work perfectly.
This could be due to incorrect assumptions of the model. 
Such cases are not interesting and we only focus on the region wherein the neural network outperforms logistic regression here. We propose a practical selective labeling:
\begin{align}
z = \begin{cases}
1, &\text{if otherwise}, \\
0, &\text{if $F(\*x) \neq \widetilde{F}(\*x)$ and $y=F(\*x)$}.
\end{cases}
\end{align}
We denote $\mathcal{Z}$ as all labels, and the practical selective labeling is referred to as selective labeling for simplicity in the rest of the paper. 
Then, at the second stage, the Difference Net is applied to learn selective labels. The structure of the method is summarized in Figure~\ref{fig:flow}. In the following lemma, we show the solution to selective labels. 
\begin{lemma} \label{lemma:solution}
Under Assumption~\ref{assump:model}, the solution to the dataset $\mathcal{D} \times \mathcal{Z}$ has the output $1-|f(\*x)-g(\*x)|$ for $\*x \in \mathcal{D}$.
\end{lemma}

\begin{figure}[H]
    \centering
    \includegraphics[scale=0.4]{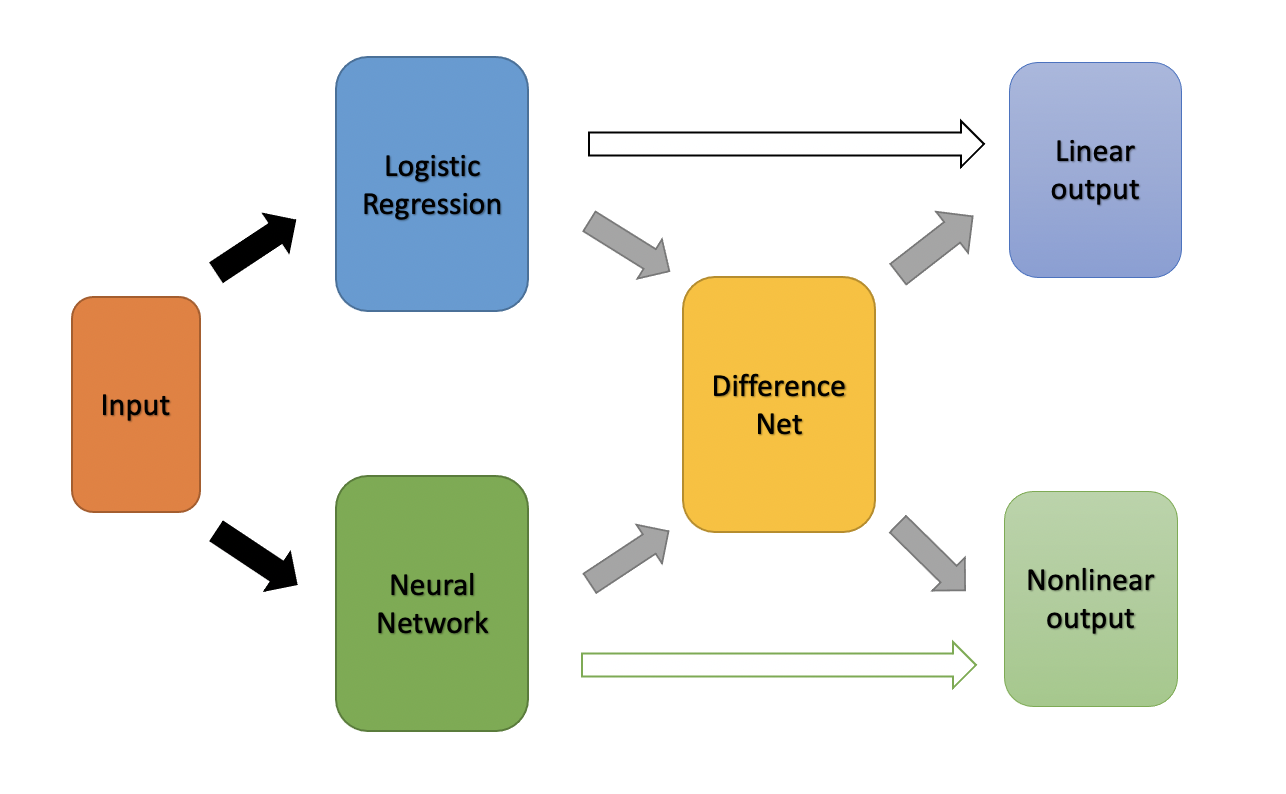}
    \caption{Two-stage Selective Learning}
    \label{fig:flow}
\end{figure}

The output to the Difference Net can also be well approximated by the neural network by Theorem~\ref{theorem:universal}, as $1-|f-g|$ is continuous given $f$ and $g$ continuous. 
As a byproduct, the rejection rate, the percentage of samples with $z=0$, can also be learned by the Difference Net \cite{ye2022}. From Lemma~\ref{lemma:solution}, we notice that if $f$ and $g$ are smooth, the new solution becomes no longer differentiable. A smooth function could be fitted by a neural network with a higher order of approximation \cite{mhaskar1996neural}. Therefore, in practice, more neurons are functioning in the Difference Net. Finally, we emphasize that the goal of the Difference Net is to provide an interpretation of neural networks, rather than to improve the accuracy.

\subsection{Explanation of the neural network}

It is important to note that different explanations \cite{lu2019good} are required for diverse situations.  The requirements of regulators and borrowers are well summarized in \cite{arrietaa2019explainable}. 
There are three types of explanations required in credit scoring, namely (i) global explanations, (ii) local instance-based explanations, and (iii) local feature-based explanations.  
Recently, sensitivity-based analysis has become increasingly popular in the interpretation of the results of neural networks \cite{horel2018,horel2020significance}. 
In this paper, we follow their analysis with slight modifications to accommodate the requirements of the credit scoring problem. 

\subsubsection{Global explanations}

% Global explanations are concerned with the performance of a model overall, instead of individual samples. Such explanations are usually preferred by regulators and managers, as they are mostly interested in the general logic behind the problem. They want to understand how well a model performs. Furthermore, they need to know why a model makes its prediction to ensure model being fair and compliant. One way to provide global explanations is to understand how each feature affects the model. 

Global explanations describe how the classification model works in general, and they interpret the logic used in its prediction. In the credit rating industry, instead of relying on individual explanations of each instance, regulators, managers, and data professionals leverage global explanations to gain an overall understanding of the scoring model in order to ensure that the model is adequate and fair in its predictions.

In this paper, the relative importance of input features is used as the global explanation. We measure the relative importance of input features at a global level to understand what has been learned by the neural network during training, and we are mainly interested in variables that lead to rejection. The global importance of input features over the training dataset of the model is defined as follows:
\begin{align}
\lambda_j = \frac{100}{C} \sqrt{\frac{1}{n} \sum_{i=1}^n \left( \frac{\partial f(\*x_i)}{\partial x^j} \right)^2},    
\end{align}
where $n$ is the number of training samples in the case of independent observations. Here, $C$ is the normalization factor such that $\sum_{j=1}^p \lambda_j = 100$, and $p$ is the number of input features. Global sensitivity is captured through partial derivatives, which are averaged across all training samples of the dataset. To avoid cancellation of positive and negative values, the square operation is used on top of the partial derivatives. This metric helps to construct a rank for features by their predictive power learned from the model. Indeed, a large value of this metric means that a large proportion of the neural network output sensitivity is explained by the considered variable. It also helps to filter out the insignificant features: a very small value of this metric means that the model outcome is almost insensitive to the feature.

% \dan{
% We are presenting a way to measure the relative importance of input features at a global level. This allows the user to understand what has been learned by the neural network during training. In this application, we are mainly interested in the question that what variables lead the dataset to rejection. The global importance of input feature over the training dataset of the model is defined as follows
% \begin{align}
% \lambda_j = \frac{100}{C} \sqrt{\frac{1}{n} \sum_{i=1}^n \left( \frac{\partial f(\*x_i)}{\partial x^j} \right)^2},    
% \end{align}
% where $n$ is the number of training samples in the case of independent observations. $C$ is the normalization factor so that $\sum_{j=1}^p \lambda_j = 100$, where $p$ is the number of input features. The derivatives are averaged across all the samples of the dataset to capture the global sensitivity and are squared to avoid the cancellation of positive with negative values.
% The main use of this metric is to rank the features by predictive power as learned by the model. Indeed, a large value of this metric means that a large proportion of the neural network output sensitivity is explained by the considered variable. It can also be used to filter out the insignificant features: a very small value of this metric means that the model outcome is almost insensitive to the feature.
% }

\begin{remark}
Categorical features are transformed into numerical variables for derivative computation. 
This is because neural networks are defined as differentiable and therefore can only handle continuous numerical features. Neural networks handle categorical features as continuous variables. With the transformation, previously described metrics can be naturally applied. 
% \wilson{please change most of passive here to active tense}To be used by a neural network model, categorical features have to be transformed into numerical variables via one-hot encoding or embedding. This is because neural networks are inherently differential and can only handle continuous and numerical features. This means that, internally, categorical features are viewed as continuous by the neural network, and hence all the previously described metrics are well suited to assess the importance of categorical input features.
\end{remark}

\subsubsection{Local instance-based explanations}

% Local instance-based explanations focus on each prediction rather than the model as a whole. Such explanations may also uncover the recurrent patterns that are present in datasets. A local instance-based explanation could provide intuition to loan officers by pointing to similar cases in support of a given decision. If other similar applications were accepted before but ended with defaulting, then loan officers will be more confident to deny the loan applicant \cite{arrietaa2019explainable}. This type of explanation is usually provided with similar samples from the dataset. 

As opposed to global explanations, local explanations provide a local understanding of predictions at instance level. Loan officers prefer such explanations because they are interested in validating whether the predictions given by the model for a loan application is justified. Loan officers review the model’s prediction by looking at other similar loan applications with the same outcome to get an understanding of why a loan application has been denied compared to other loan applications that were previously accepted and then ended up defaulting \cite{arrietaa2019explainable}. This type of explanation is usually provided in the form of prototypes (i.e., categorizing the applications based on similarity).

Our Difference Net serves as the perfect tool for the local instance-based explanations. The logistic regression has provided good baseline accuracy for the dataset. While the neural network could further improve the logistic regression, the improvement is local. With the Difference Net, we are able to identify the local region where the neural network has significant improvements.  With the localization, recurrent patterns can be found from data. Furthermore, the output from the logistic regression is combined with important features of samples to provide visualization. As a result, a typical explanation to use a neural network for a loan officer may resemble the following:
``This person is rejected because his/her repayment in  the last month has been delayed for 2 months, which is similar to A and B in the datasets. Although they have good credit scores calculated by the traditional method, A and B could not pay off."
Furthermore, we are able to provide the theoretical justification of the generalization error, which is discussed later.

\subsubsection{Local feature-based explanations}

Local feature-based explanations are concerned with how and why a specific prediction is made, at a local level. Such explanations are usually preferred by borrowers, as they are most interested in why their applications are denied. More importantly, this information can help them to improve their credit scores to obtain approval for loans in the future. If reasonable explanations are provided, they can work on their deficiencies to obtain better credit scores. Feature relevance scores or specific rules can be provided for these explanations. In terms of our selective learning, it is also beneficial to understand why certain samples are rejected by logistic regression.

A standard method to understand the local behaviors of a differentiable function is through the Taylor expansion. The Taylor series focuses on a small neighborhood of one sample of interest and provides a good approximation locally. Here, it can serve as a useful tool to capture the local relative importance of input features. In practice, a first-order Taylor expansion is usually sufficient for the analysis. Mathematically, for any input vector $\*x$ close to $\*x_0$, we have
\begin{align}
f(\*x) - f(\*x_0) = (\*x - \*x_0)^T \nabla f(\*x_0) 
+ o(\|\*x-\*x_0\|)
\end{align}
for $\*x \rightarrow \*x_0$. The Taylor expansion shows that neural network output $f(\*x)$ can be well explained by its gradient locally. Then, for a sample $\*x_0$ with input feature $j$, it is sufficient to look at its partial derivative as local importance:
\begin{align}
\lambda^j_0 =  \frac{\partial f(\*x_0)}{\partial x^j}.
\end{align}
For categorical variables, we simply consider modifying the variable to the nearest value, i.e., 
\begin{align}
\lambda^j_0 = f(x_1, \dots, x_j \pm 1, \dots, x_p) - f(x_1, \dots, x_j, \dots, x_p).
\end{align}
If the difference is significant, then the feature is important locally. Intuitively, this informs us why such a data point belongs to the rejected set. These explanations can be provided by the neural network.

% \begin{remark}
% The methodology is only applied to independent variables, which can be altered by customers, such as the late payment. Variables, such as education, are not considered here, as it is difficult to modify it in the short period of time. 
% \end{remark}

\subsection{Concentration of measure results}

We have designed a learning algorithm that can identify the rejection region of the logistic regression in the earlier section. 
As a byproduct, it is interesting to estimate the data rejection rate, i.e., the percentage of samples that are rejected.  
We provide the definition of the rejection rate. 
\begin{definition}
Assume that the training and test data follow from a universal underlying distribution with probability density function $\eta(\*x)$. Denote $\Omega_d$ and $\Omega_{nd}$ as the default and not default regions of the rejection region $\Omega$, respectively. The rejection rate is defined as the percentage of rejected data:
\begin{align}
\gamma =   \int_{\Omega_d} \eta(\*x) (1-p(\*x)) \ d\*x +  \int_{\Omega_{nd}} \eta(\*x) p(\*x) \ d\*x.
\end{align}
\end{definition}

We provide some theoretical justification of the data rejection rate.
We require the rejection rate in the training set to serve as a natural and rigorous estimation of the intrinsic value for the population.
This is called generalization, a key topic in machine learning.
Generalization measures how accurately an algorithm can predict output values for unseen data. 
As an intrinsic value of an unknown dataset, the rejection rate provides us with important information about the non-linearity of the data and should be estimated accurately from outside of the sample.
%Please ensure intended meaning is retained above.
In this section, we study the generalization of the rejection rate.
We provide a rigorous framework that uses the concentration of measures to estimate the rejection rate in testing.

% We present the concentration of measure in this study.
For a sample $s$, if $\mathcal{I}$ is an indicator function and $\mathcal{I}_{s}$ indicates whether sample $s$ is rejected or not, then
\vspace{-0.2cm}  
\begin{align}
    \mathcal{I}_{s} = 
    \begin{cases}
        1, & \text{with probability}\,\, \gamma\,; \\
        0, & \text{with probability}\,\, 1-\gamma\,.
    \end{cases}
\end{align}
Let us say that $\mathcal{X}$ and $\mathcal{W}$ are the training and test sets. 
Denote by $n_{\mathcal{X}}$ and $n_{\mathcal{W}}$ the size of $\mathcal{X}$ and $\mathcal{W}$.
Then, for $\{\mathbf{x}_i\}_{i=1}^{n_{\mathcal{X}}} \in \mathcal{X}$, we can write the training set rejection rate $\gamma_{\mathcal{X}}$ as
$\gamma_{\mathcal{X}} = \frac{1}{n_{\mathcal{X}}} \sum_{i=1}^{n_{\mathcal{X}}} \mathcal{I}_{\mathbf{x}_i} $,
% \begin{align}
%     \label{eqn:train-corruption}
%     \gamma_{\mathcal{X}} = \frac{1}{n_{\mathcal{X}}} \sum_{i=1}^{n_{\mathcal{X}}} \mathcal{I}_{\mathbf{x}_i} 
% \end{align}
where $\mathcal{I}_{\mathbf{x}_i}$ is the rejection indicator function.
Similarly, for $\{\mathbf{w}_i\}_{i=1}^{n_{\mathcal{W}}} \in \mathcal{W}$, the test set rejection rate $\gamma_{\mathcal{W}}$ is $\label{eqn:test-corruption}
\gamma_{\mathcal{W}} = \frac{1}{n_{\mathcal{W}}} \sum_{i=1}^{n_{\mathcal{W}}} \mathcal{I}_{\mathbf{w}_i}$
% \begin{align}
%     \label{eqn:test-corruption}
%     \gamma_{\mathcal{Z}} = \frac{1}{n_{\mathcal{Z}}} \sum_{i=1}^{n_{\mathcal{Z}}} \mathcal{I}_{\mathbf{z}_i} 
% \end{align}
where $\mathcal{I}_{\mathbf{w}_i}$ is the rejection indicator function.
The following lemma indicates that the training rejection rate is close to the population rejection rate with high probability.

\begin{lemma}[Concentration result between the training set and universal]
Assume that both $ \mathcal{X}$ and  $ \mathcal{W}$ are from some universal distribution where each sample is rejected with probability $\gamma$. 
% Denote $c_{\mathcal{S}}$ as the corruption threshold of the sample subset with replacement, $\mathcal{S} \subset \mathcal{X}$.
% Assume training data size is $n_{\mathcal{X}}$.
% and subset size is $N_{\mathcal{S}}$. 
The inequality is satisfied for any positive $\epsilon_1$:
\begin{equation}
    \label{eqn:hoeffding-selective-net}
    \mathcal{P}(\left\| \gamma_{\mathcal{X}} - \gamma \right\| \geq \epsilon_1)\leq 2 e^{-2n_{\mathcal{X}} \epsilon_1^2} \,.
\end{equation}
\label{lemma:concentration-sample-train}
\end{lemma}
Then, we need a concentration of measure results between the training and test set for generalization purposes.

\begin{lemma}
[Concentration between the training and test set]
% We have the continuation of the above concentration of measure lemma. 
Assume that both $ \mathcal{X}$ and $\mathcal{W}$ are from some universal distribution where each sample is rejected with probability $\gamma$.
% Assume training and test data come from a universal system that has a fixed corruption rate (noise rate). 
% Let $c_{\mathcal{Z}}$ be the corruption fraction in the test data and the test dataset size is $n_{\mathcal{Z}}$.
Then, for any positive pair $(\epsilon_1, \epsilon_2)$, we have the following inequality:
\begin{equation}
    \mathcal{P}( \| \gamma_{\mathcal{X}} - \gamma_{\mathcal{W}}\| \ge \epsilon_1 +\epsilon_2) \le 2 e^{-2n_{\mathcal{X}} \epsilon^2_1} + 2 e^{-2n_{\mathcal{W}} \epsilon^2_2 }  \,.
\end{equation}
\label{lemma: concentration-sample-test}

\end{lemma}

From the above lemmas, we can see that the rejection rate for the training set could be a good estimator of the rejection threshold of the population and test set.

\section{Empirical results}

While there have been many datasets used in the benchmark work \cite{lessmann2015benchmarking}, most of them are either not publicly available, lack variable names, or contain an insufficient number of samples which does not guarantee statistically significant results. In our study, we focus on two publicly available datasets, which have variable names that are easier to interpret, and also have enough samples such that our estimation of the rejection rate is statistical significance. In both datasets, the setup of neural networks and the optimization algorithms in backpropagation are the same. At the first learning stage, we train a neural network on the datasets and compare them with the conventional explainable logistic regression. The architecture is rather simple for credit scoring problems. Based on the model architecture optimization, we observe that one hidden layer of two units with the logistic activation function is sufficient. At the second learning stage, we construct a Difference Net and find that one hidden layer of five neurons is sufficient. The additional neurons are required as the output function becomes non-differentiable, as pointed out in  Lemma~\ref{lemma:solution}. The backpropagation optimization is solved by the conjugate gradient method. We limit the total number of training epochs to 500 as the errors decay slowly.

We then study the model performance within each learning stage. The neural network is compared with the logistic regression within the first learning stage. As discussed in \cite{lessmann2015benchmarking}, there is no single perfect performance measurement for the credit scoring problem. In this study, three types of measurements are considered: (i) classification error — this tracks the percentage of the corrected prediction and is the most intuitive metric to evaluate the model fit; (ii) receiver operating characteristic curve (ROC) and area under the curve (AUC) — they provide a comprehensive evaluation of classification models.  
(iii) confusion matrix — in the credit scoring, it is important to correctly predict actual default cases, because false negatives may carry huge financial cost to banks. 
Although using other measurements that could potentially provide different insights into the model fit  \cite{lessmann2015benchmarking}. We believe classification error, ROC and AUC, and confusion matrix are sufficient since our primary concern of this study is model interpretation instead of model accuracy. We rely on the confusion matrix to further break down the default predictions by actual and predicted conditions.
After comparing the global model performance at the first stage, we focus on analyzing performance of the Difference Net. As mentioned in Section 2.2.2, the Difference Net serves as a natural localized tool to find the rejection set which mostly contains predicted defaults by neural network. As AUC and ROC are not appropriate for evaluating highly skewed predictions, only the classification error is applicable to the rejected set for interpretation purpose. While the logistic regression might perform similarly to neural networks at the global level, there exists a significant difference at the local rejection region. Local evaluations could provide a more specific understanding of the model performance.

\subsection{Taiwan data}

\subsubsection{Description of data}
We first choose the Taiwan credit score dataset \footnote{https://archive.ics.uci.edu/ml/datasets/default+of+credit+card+clients}. The payment data in October 2015 was collected from an important bank (a cash and credit card issuer) in Taiwan and the targets were credit card holders of the bank.
 Among the total 30,000 observations, 6,639 (22.12$\%$) relate to the cardholders with default payments. This research employs a binary variable, default indicator, as the response variable.  The dataset is randomly partitioned into the following sets: $75\%$ training set and $25\%$ test set. In addition, this dataset contains the following 23 variables as explanatory variables:
\begin{itemize}
\item $x_1$: Amount of the given credit (NT dollar): this includes both the individual consumer's credit and his/her family (supplementary) credit.
\item $x_2$: Gender (1 = male; 2 = female).
\item $x_3$: Education (1 = graduate school; 2 = university; 3 = high school; 0,4,5,6 = others).
\item $x_4$: Martial status (1 = married; 2 = single; 3 = divorce; 0 = others). 
\item $x_5$: Age (years).
\item $x_6 - x_{11}$: History of past payments. These variables track the past monthly payment records (from April to September, 2005) as follows: $x_6$ = repayment status in September, 2005; $x_7$= the repayment status in August, 2005; \dots; $x_{11}$ = repayment status in April, 2005. The measurement scale for the repayment status is: $-2$ = No consumption; $-1$ = Paid in full; 0=The use of revolving credit; 1=Payment delay for one month; 2 = Payment delay for two months; \dots; 8=Payment delay for eight months; 9 = Payment delay for nine months and above. 
\item $x_{12} - x_{17}$: Amount of bill statement (NT dollar). $x_{12}$ = Amount of bill statement in September, 2005; $x_{13}$ = Amount of bill statement in August, 2005; \dots; $x_{17}$ = Amount of bill statement in April, 2005. 
\item $x_{18}-x_{23}$: Amount of previous payment (NT dollar). $x_{18}$ = Amount paid in September, 2005; $x_{19}$ = Amount paid in August, 2005; $x_{23}$ = Amount paid in April, 2005. 
\item $y$: Client's behavior; 0 = Not default; 1 = Default.
\end{itemize}

\subsubsection{Results}

At the first stage, we compare the neural network and the logistic regression. The classification test error is $18.2\%$ for the logistic regression and $17.7\%$ for the neural network. This result is consistent with the literature \cite{yeh2009comparisons}. The result indicates that, overall, the neural network has a slight improvement over the logistic regression. This is not surprising as the datasets are imbalanced here given that most of the clients are creditable. The logistic regression has already provided a good approximation to identify the non-default cases and therefore is very popular in practice. In addition, the AUC is $71.1\%$ for the logistic regression and $73.1\%$ for the neural network. The comparison of the ROC curve for the two methods is plotted in Figure~\ref{fig:AUC}. The neural network has a consistent improvement over the logistic regression. Next, we look at confusion matrix for further analyses. The confusion matrix for the logistic regression is shown in  Table~\ref{tab:LR_confusion}. 
The confusion matrix for the neural network is given in Table~\ref{tab:NN_confusion}. 
These matrices show that the neural network has captured more relevant default cases than logistic regression, which is measured by recall in machine learning. More specifically, the recall of the logistic regression is 24.9$\%$, whereas the neural network achieves a recall of 31.3$\%$. This is a significant improvement because a higher recall means more actual defaults are retrieved by the model. Overall, these results show that the neural network has outperformed logistic regression in terms of accuracy with a higher recall for this dataset, and it is worth exploring where these improvements are made.

\begin{figure}[h]
\centering
\includegraphics[scale=0.5]{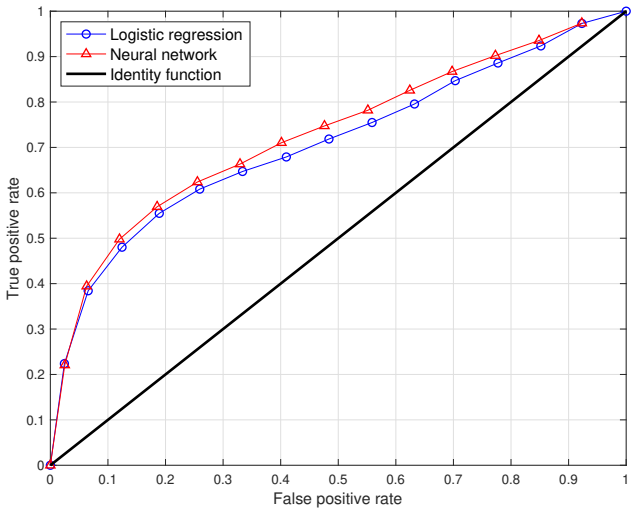}
\caption{ROC curves of the logistic regression and the neural network of the Taiwan credit score dataset}
\label{fig:AUC}
\end{figure}

\begin{table}[h]
    \centering
    \begin{tabular}{ccc}
    \hline
     $n=7500$  & Predicted: Default & Predicted: Not default  \\ \hline
Actual: Default  & 397  & 1198  \\ \hline
Actual: Not default  & 168  & 5737  \\ \hline
    \end{tabular}
    \caption{Confusion matrix of the logistic regression for the Taiwan credit score dataset.}
    \label{tab:LR_confusion}
\end{table}

\begin{table}[h]
    \centering
    \begin{tabular}{ccc}
    \hline
     $n=7500$  & Predicted: Default & Predicted: Not default  \\ \hline
Actual: Default  & 500  & 1095  \\ \hline
Actual: Not default  & 234  & 5671  \\ \hline
    \end{tabular}
    \caption{Confusion matrix of the neural network for the Taiwan credit score dataset. }
    \label{tab:NN_confusion}
\end{table}

At the second stage, we apply the Difference Net to the dataset. The classification test error is 1.6$\%$, indicating its predictive power. The percentage of predicted selected samples in the test set is 2.3$\%$, implying that, for most of the dataset, the logistic regression is sufficient. Therefore, we can have good confidence in the logistic regression in most cases. However, for approximately 2$\%$ - 3$\%$ of the dataset, the neural network strongly disagrees with the logistic regression. Special attention must be paid to this case. Conditioned on the rejected set, the classification test error is 63.8$\%$ for the logistic regression and 34.5$\%$ for the neural network. This indicates that, for the rejected set, the neural network has tremendously outperformed the logistic regression and hence should be adopted.

Further analyses are applied to the Difference Net. Among the rejected set, all samples are predicted as default by the neural network. At the global level, the Difference Net identifies samples where risks are underestimated by the logistic regression. 92.7$\%$ of rejected data has a payment delayed for two months in September 2015 (i.e., $x_6$ = 2). This pattern suggests that samples are rejected mainly due to the variable $x_6$, and further feature importance is no longer needed. For local feature-based explanations, we check samples with $x_6$ = 2 and find that the averaged difference of f(x) is 64.5$\%$ by modifying $x_6$ = 1, showing that $x_6$ plays an essential role in the rejected set. For local instance-based explanations, we combine the output from the logistic regression variable $x_6$.  In Figure~\ref{fig:local_explain}, we show the rejected set with respect to the logistic regression output and $x_6$. The result clearly indicates that, for overall lower-risk customers determined by logistic regression, if the $x_6$ = 2, then its risk is significantly underestimated, as pointed out by the neural network. This could provide a convincing explanation to local officers.

% We focus on this variable for deeper analysis. We plot the probability of default with respect to it and the corresponding logit function in Figure~\ref{fig:X6}. We group all cases that $x_6 \geq 3$ together as there is not sufficient data. Clearly, there are some nonlinearities for $x_6$. First of all,  it is not clear where to put the $x_6=-2$ as it didn't contain too much information. Second, there is a diminishing marginal effect: the impact of the late payment becomes less severe for long history. Such nonlinearity can be much better handled by the neural network. 

\begin{figure}[h]
    \centering
    \includegraphics[scale=0.5]{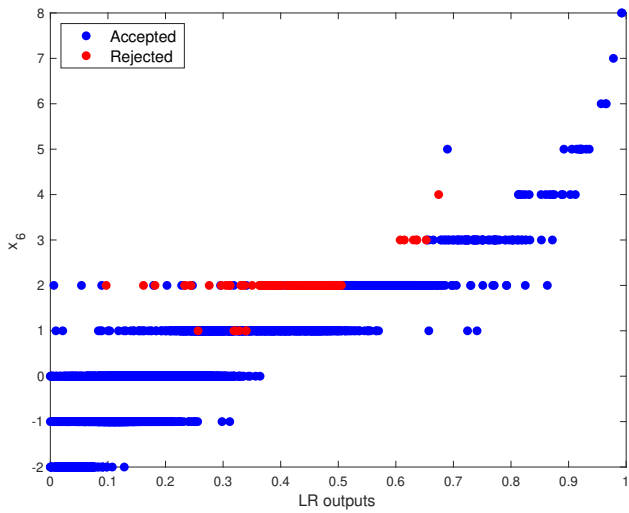}
    \caption{Predicted rejected samples for test data with respect to the logistic regression output and $x_6$ of the Taiwan credit score dataset.}
    \label{fig:local_explain}
\end{figure}

% \begin{figure}
% \centering
% \includegraphics[scale=0.5]{figures/X6}
% \caption{The probability of default and logit function for the $6^{\text{th}}$ variable}
% \label{fig:X6}
% \end{figure}

\subsection{Kaggle dataset---``Give me some credit"}

\subsubsection{Data description}

We also test the Kaggle credit score dataset \footnote{https://www.kaggle.com/c/GiveMeSomeCredit/overview}. For simplicity, data with missing variables are removed. The potential accuracy can certainly be improved with the appropriate consideration of samples with missing variables, but is not the primary concern of this paper. Among the total 120969 observations, 8,357 (6.95$\%$) relate to the cardholders with default payments. This indicates that the data are seriously imbalanced. Similarly, the dataset is randomly partitioned into $75\%$ training and $25\%$ test sets. The dataset contains 10 variables as explanatory variables:
\begin{itemize}
\item $x_1$: Total balance on credit cards and personal lines of credit except real estate and no installment debt such as car loans divided by the sum of credit limits (percentage).
\item $x_2$: Age of borrower in years (integer).
\item $x_3$: Number of times borrower has been 30 - 59 days past due but no worse in the last 2 years (integer).
\item $x_4$: Monthly debt payments, alimony, living costs divided by monthly gross income (percentage). 
\item $x_5$: Monthly income (real).
\item $x_6$: Number of open loans (installments such as car loan or mortgage) and lines of credit (e.g., credit cards) (integer). 
\item $x_7$: Number of times borrower has been 90 days or more past due (integer). 
\item $x_8$: Number of mortgage and real estate loans including home equity lines of credit (integer). 
\item $x_9$: Number of times borrower has been 60 - 89 days past due but no worse in the last 2 years (integer). 
\item $x_{10}$: Number of dependents in family, excluding themselves (spouse, children, etc.) (integer). 
\item $y$: Client's behavior; 1 = Person experienced 90 days past due delinquency or worse.
\end{itemize}

\subsubsection{Results}

At the first stage, classification test errors are $7.1\%$ for the logistic regression and $6.8\%$ for the neural network. As the non-default cases cover approximately $7\%$ of the dataset, the classification errors are not informative. In addition, the AUC is $69.4\%$ for the logistic regression and $81.6\%$ for the neural network, indicating that the neural network outperforms the logistic regression. The comparison of the ROC curve for the two methods is given in Figure~\ref{fig:AUC_GMSC}. These results are consistent with the results from Kaggle leaderboard. 
The confusion matrices for the logistic regression and neural network are provided in Table~\ref{tab:LR_confusion_GMSC} and Table~\ref{tab:NN_confusion_GMSC}, respectively. Recall rates are 18.0$\%$ for the neural network compared to 3.1$\%$ for the logistic regression. These show that the neural network has a much better ability to capture the actual default cases. Thus, we are comfortable to conclude that the neural network has better predictive power in this dataset. 

\begin{figure}[h]
\centering
\includegraphics[scale=0.5]{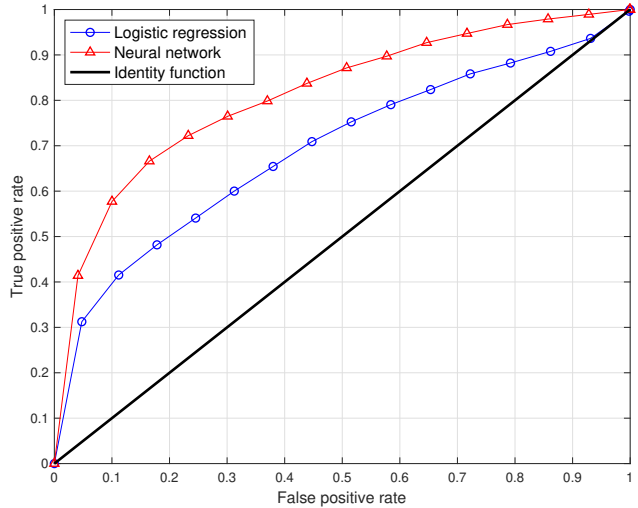}
\caption{ROC curves of the logistic regression and the neural network of the Kaggle dataset.}
\label{fig:AUC_GMSC}
\end{figure}

\begin{table}[h]
    \centering
    \begin{tabular}{ccc}
    \hline
     $n=30067$  & Predicted: Default & Predicted: Not default  \\ \hline
Actual: Default  & 66  & 2072  \\ \hline
Actual: Not default  & 50  & 27879  \\ \hline
    \end{tabular}
    \caption{Confusion matrix for the logistic regression of the Kaggle dataset.}
    \label{tab:LR_confusion_GMSC}
\end{table}

\begin{table}[h]
    \centering
    \begin{tabular}{ccc}
    \hline
     $n=30067$  & Predicted: Default & Predicted: Not default  \\ \hline
Actual: Default  & 384  & 1754  \\ \hline
Actual: Not default  & 287  & 27642  \\ \hline
    \end{tabular}
    \caption{Confusion matrix for the neural network of the Kaggle dataset.}
    \label{tab:NN_confusion_GMSC}
\end{table}

% At the second stage, we apply the Difference Net to the dataset. The classification test error is $0.9\%$, indicating the high predictive power of the Difference Net. The percentage of predicted selected samples in the test set is $1.4\%$. Among the rejected set, 99.7$\%$ of disagreement comes from  predictions that the neural network predicts default, whereas the logistic predicts non-default. Conditioned on the rejected set, the error is 59.4$\%$ for the logistic regression and 40.4$\%$ for the neural network. This indicates that, for the rejected set, the neural network has much better performance than the logistic regression and should be adopted \wilson{high similarity with the last experiment. Tried to modify, but issue still exists}. 

At the second stage, we apply the Difference Net to the dataset. The Difference Net identifies $1.4\%$ of selected samples and only has $0.9\%$ classification test error. Among the rejected set, 99.7$\%$ of disagreement comes from predictions that the neural network predicts default, whereas the logistic regression predicts non-default. Among the rejected set, the neural network shows tremendous improvements over the logistic regression, with $40.4\%$ classification test error compared to $59.4\%$ for the logistic regression. The $19\%$ improvement suggests that neural network should be applied to the rejected set. 

For global explanations, the feature importance of the Selected Net based on sensitivity analysis is given in Figure~\ref{fig:global_explain_GMSC}. Note that the variables $x_7$, $x_9$, and $x_3$ are the most important ones. Cumulatively, they explain $80\%$ of the variance. Therefore, we focus on these variables. 

\begin{figure}[h]
    \centering
    \includegraphics[scale=0.5]{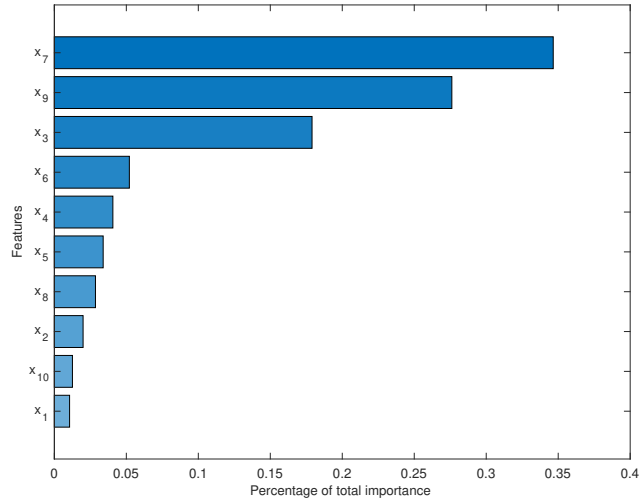}
    \caption{Global features importance of Difference Net of the Kaggle dataset.}
    \label{fig:global_explain_GMSC}
\end{figure}

For local instance-based explanations, we combine the output from the logistic regression variable. In Figure~\ref{fig:local_explain_GMSC}, we show the rejected set with respect to the logistic regression output and $x_7, x_9, x_3$ separately. The result clearly indicates that, for overall lower-risk customers determined by logistic regression, if these variables take large values, then the overall risk is significantly underestimated, as pointed out by the neural network. Therefore, recommendations to borrowers will be to reduce the number of past dues as represented by these variables. To understand why these variables fail to be captured by the logistic regression, we plot their logit functions in1 Figure~\ref{fig:nonlinearity}. We observe diminishing marginal effects. Taking $x_7$ as an example, there will be a huge difference in the probability between cases that the payment is never past due once. However, for these past due cases, whether it is four or five times has little difference. This effect introduces non-linearity to the model and therefore cannot be explained by the logistic regression.
Overall, qualitatively, we obtain results similar to the Taiwan dataset. 

\begin{figure}[h]
    \centering
    \includegraphics[scale=0.6]{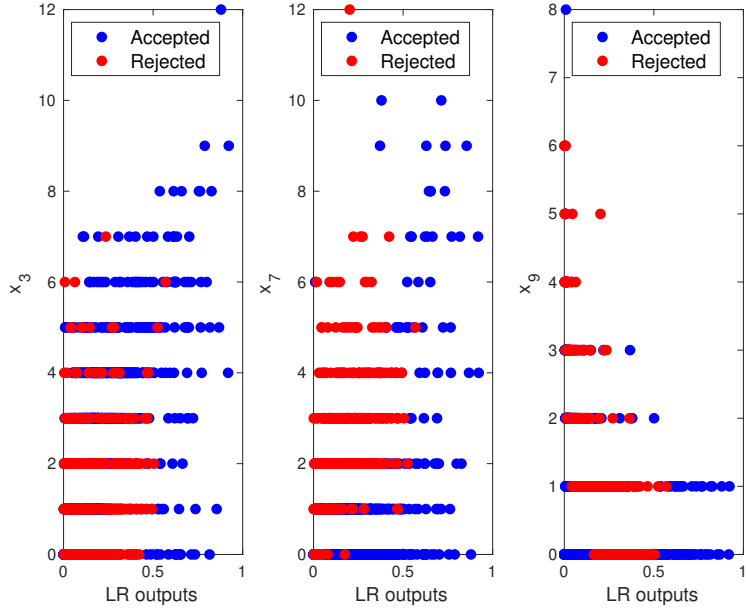}
    \caption{Predicted rejected samples for test data with respect to the logistic regression output and $x_3, x_7, x_9$ of the Kaggle dataset.}
    \label{fig:local_explain_GMSC}
\end{figure}

\begin{figure}[h]
    \centering
    \includegraphics[scale=0.6]{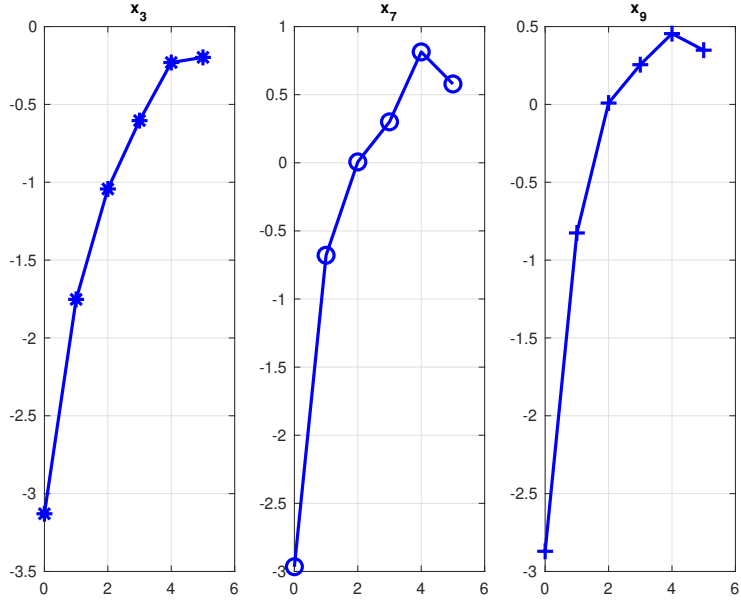}
    \caption{Non-linearities of logit functions of $x_3,x_7,x_9$ in the Kaggle dataset.}
    \label{fig:nonlinearity}
\end{figure}

\section{Conclusions}

In this paper, we study the usage of the machine learning method in credit scoring. 
As introduced, some complicated non-linear machine learning methods have better predictive power; however, they are considered black-box structures without the good interpretability required by financial regulators. 
As a consequence, they have not been widely adopted in credit scoring.
To resolve this issue, we introduce a neural network with the selective option to distinguish whether the datasets can be explained by the linear models or not.
For the portion of the datasets that cannot explained by the linear model well, our learning model can feed the data into more complex machine learning methods.
According to our model, we observe that, for most of the datasets, the logistic regression will be sufficient to achieve reasonably good accuracy; meanwhile, for some specific data portions, a shallow neural network model leads to much better accuracy without a significance sacrifice in interpretability.
We show that machine learning does have better predictability than naive logistic regression. However, there is a compromise wherein the black-box machine learning method would lose its interpretability. Therefore, practitioners have been hesitant to adopt them. We propose a novel Selective Net, which can identify the data where the simple logistic regression fails. We show that, for most of the dataset, the logistic regression has been very useful. There is only a small amount of the dataset where the logistic regression would fail. Using the Difference Net, it is recommended that practitioners should still use the logistic regression for most cases, but they should switch to the neural network for specific regions. 
% For further studying, proposing a structured shallow structure of neural networks based on the interpretations of Selective Learning could be a future direction. 
% This step could maintain the accuracy of black-box machine learning methods and interpretation of logistic regression.
For future study, one possible direction is to generalize the selective learning framework to other complicated machine learning methods.
Another potential direction of research would be to extend our method to other finance applications, including fraud detection and anti-money laundering.
In these areas, interpretability is essential.

% \section*{Acknowledgement}
% The authors would like to thank  for their insight and comments.
%Please check if a name is missing above.

\section*{Supplementary}

\begin{proof}[Proof of Lemma~\ref{lemma:solution}]

As the neural network has the true output $f(\*x)$, the probability of the difference between the logistic regression and the neural network is $|f(\*x)-g(\*x)|$. Therefore, the acceptance rate at $\*x$ is $1-|f(\*x)-g(\*x)|$.

\end{proof}

\begin{proof}[Proof of Lemma ~\ref{lemma:concentration-sample-train}]  
From Hoeffding's inequality, for any $\epsilon_1 > 0$, 
\begin{align}
\label{eqn:selective-net-dummy01}
    \mathcal{P}\left(\left| \frac{1}{n_{\mathcal{X}}}\sum_{i=1}^{n_{\mathcal{X}}}\mathcal{I}_{\mathbf{x}_i} - \mathbb{E}\left[\frac{1}{n_{\mathcal{X}}}\sum_{i=1}^{n_{\mathcal{X}}}\mathcal{I}_{\mathbf{x}_i}\right] \right| \geq \epsilon_1 \right) \leq 2 e^{-2n_{\mathcal{X}} \epsilon_1^2}
\end{align}
where $\mathbb{E}\left[\frac{1}{n_{\mathcal{X}}}\sum_{i=1}^{n_{\mathcal{X}}}\mathcal{I}_{\mathbf{x}_i}\right]$ is equal to the universal rejected rate $\gamma$. 
From the definition of $\gamma_{\mathcal{X}}$ in the paper, (\ref{eqn:selective-net-dummy01}) is equivalent to
\begin{align}
\label{eqn:selective-net-dummy02}
    \mathcal{P}(\left\| \gamma_{\mathcal{X}} - \gamma \right\| \geq \epsilon_1)\leq 2 e^{-2n_{\mathcal{X}} \epsilon_1^2}
\end{align}
This gives the concentration between training $\gamma_{\mathcal{X}}$ and universal $\gamma$.

\end{proof}

\begin{proof}[Proof of Lemma~\ref{lemma: concentration-sample-test}]
Same as the proof of Lemma \ref{lemma:concentration-sample-train}, applying Hoeffding's inequality on $\mathcal{W}$ and with the definition of $\gamma_{\mathcal{W}}$ in the paper; then, for any $\epsilon_2 \ge 0$,
\begin{equation}
    \label{eqn:selective-net-dummy04}
    \mathcal{P}( \| \gamma - \gamma_{\mathcal{W}}\| \ge \epsilon_2) \le 2 e^{-2n_{\mathcal{W}}\epsilon_2^2 }
\end{equation}
Apply the union bound on (\ref{eqn:selective-net-dummy02}) and (\ref{eqn:selective-net-dummy04});
then, with probability of more than $1-2 e^{-2n_{\mathcal{X}}\epsilon_1^2} -2e^{-2n_{\mathcal{Z}}\epsilon_2^2 }$, both of the following inequalities satisfy at the same time
\begin{align*}
    & \| \gamma_{\mathcal{X}} - \gamma \| < \epsilon_1,\\
    & \| \gamma - \gamma_{\mathcal{W}} \| < \epsilon_2.
\end{align*}
Use the triangle inequality, then
\begin{align*}
    \| \gamma_{\mathcal{X}} - \gamma_{\mathcal{W}}  \| \leq  \| \gamma_{\mathcal{X}} - \gamma \| + \| \gamma-\gamma_{\mathcal{W}} \| < \epsilon_1+\epsilon_2
\end{align*}
with probability more than $1-2 e^{-2n_{\mathcal{X}}\epsilon_1^2} -2e^{-2n_{\mathcal{W}}\epsilon_2^2 }$ as desired.

\end{proof}

%%%% Bibliography  %%%%%%%%%%
\bibliography{chen}
\bibliographystyle{plain}

\end{document}